\documentclass[a4paper,11pt]{article}

\usepackage[pdftex,paper=a4paper,left=3.2cm,right=3.2cm,top=4cm,bottom=4cm]{geometry}
\usepackage{amsmath,amssymb,amsthm}
\usepackage{ifpdf}

\ifpdf
  \usepackage[pdftex,linktocpage,
    pdfauthor={Patrick Briest, Paul Goldberg, Heiko Roeglin},
    pdftitle={Approximate Equilibria in Games with Few Players},
    breaklinks]{hyperref}
\else
  \usepackage[breaklinks]{hyperref}
\fi

\newtheorem{theorem}{Theorem}

\newtheorem{corollary}[theorem]{Corollary}
\newtheorem{definition}[theorem]{Definition}

\def\numplayers{k}
\def\supsize{t}

\renewcommand{\Pr}[1]{\mbox{\rm\bf Pr}\left[#1\right]}

\title{Approximate Equilibria in Games with Few Players\thanks{This 
        work was supported by DFG grant VO 889/2, EPSRC Grant
        GR/T07343/02, and by the EU within the 6th Framework
        Programme under contract 001907 (DELIS).}} 
	    
\author{Patrick Briest\\
       Dept.\ of Computer Science\\
       University of Liverpool, U.K.\\
       \href{mailto:Patrick.Briest@liverpool.ac.uk}{Patrick.Briest@liverpool.ac.uk}
\and
       Paul W.\ Goldberg\\
       Dept.\ of Computer Science\\
       University of Liverpool, U.K.\\
       \href{mailto:P.W.Goldberg@liverpool.ac.uk}{P.W.Goldberg@liverpool.ac.uk}
\and
       Heiko R\"oglin\\
       Dept.\ of Computer Science\\
       RWTH Aachen, Germany\\
       \href{mailto:roeglin@cs.rwth-aachen.de}{roeglin@cs.rwth-aachen.de}
}
        
\begin{document}

\maketitle

\begin{abstract}
We study the problem of computing approximate Nash equilibria
($\epsilon$-Nash equilibria) in normal form games, where the number of
players is a small constant. We consider the approach of looking for
solutions with constant support size. It is known from recent work
that in the 2-player case, a $\frac{1}{2}$-Nash equilibrium can be
easily found, but in general one cannot achieve a smaller value of
$\epsilon$ than $\frac{1}{2}$. In this paper we extend those results
to the $k$-player case, and find that $\epsilon = 1-\frac{1}{k}$ is
feasible, but cannot be improved upon. We show how stronger
results for the 2-player case may be used in order to slightly improve
upon the $\epsilon = 1-\frac{1}{k}$ obtained in the $k$-player case.
\end{abstract}

\section{Introduction}
\label{sec:intro}

A game in normal form has $\numplayers$ players, and for each player
$p$ a set $S^p$ of pure strategies. In this paper we assume that all
sets $S^p$ are of the same size $n$. The set $S$ of pure strategy
profiles is the cartesian product of the sets $S^p$. For each player
$p$ and each $s\in S$, the game has an associated value $u^p_s$, being
the utility or payoff to player $p$ if all players choose $s$.  Note
that the number of quantities needed to specify the game is
$\numplayers\cdot n^\numplayers$, so we must take $\numplayers$ to be
a constant, for the game's description to be of size polynomial
in $n$.

A mixed strategy for player $p$ is a probability distribution over
$S^p$. Suppose that each player $p$ chooses distribution $D^p$ over
$S^p$. Let us define a player's {\em regret} to be the highest payoff
he could obtain by choosing a best response to the other players'
mixed strategies, minus his actual expected payoff. Then, a {\em Nash
equilibrium} is a set of $D^p$'s for which all players' regrets are
zero. We say that the distributions $D^p$ form an {\em $\epsilon$-Nash
equilibrium} provided that all players' regrets are at most
$\epsilon$.

\subsection{Recent work}

Due to the apparent difficulty of computing Nash equilibria
exactly~\cite{CD2,CDT,DGP}, recent work has addressed the question of
polynomial-time computability of approximate Nash equilibria, in
particular the $\epsilon$-Nash equilibria defined above. The effort
has mostly addressed 2-player games, and the general question is, for
which values of $\epsilon$ can $\epsilon$-Nash equilibria be found in
polynomial time? The payoffs $u^p_s$ in the games are restricted to
lie in the range $[0,1]$, since otherwise the payoffs (and associated
values of $\epsilon$) could be rescaled arbitrarily. Recent papers
include~\cite{BBM07,TS07}, which show how to efficiently compute
2-player $\epsilon$-Nash equilibria for $\epsilon=0.36392$ and
$0.3393$ respectively. We do not know of similar work so far that
addresses the $k$-player case considered here.

The {\em support} of a probability distribution is the number of
elements of its domain that have non-zero probability.  Solutions of
games having small support (that is, players' mixed strategies do not
give positive probability to many of the pure strategies) are
attractive for two reasons.  From the perspective of modelling a
plausible outcome, we expect a participant in a game to prefer a
simple behaviour. Also, if constant-support solutions exist, then if
the number of players is constant, we can find them in polynomial time
by brute force.  For example, this approach is used by B\'ar\'any et
al.~\cite{BVV} to find Nash equilibria in random games. Our results
show that additional players makes it genuinely more
intractable to find a satisfactory solution; if we restrict the
support sizes to any constant then the worst-case regret of
some player increases.

It is known that for 2 players, if we restrict ourselves
to solutions with constant support, then there is a lower bound of
$\frac{1}{2}$ on the best $\epsilon$ that can be
achieved~\cite{FNS07}, and this lower bound is achieved by a very
simple algorithm of Daskalakis, Mehta and Papadimitriou~\cite{DMP06}.

Regarding the issue of how large the support size needs to be in order
to allow an $\epsilon$-Nash equilibrium, it is shown by
Alth\"ofer~\cite{A94} and Lipton et al.~\cite{LMM03} that
$\log(n)/\epsilon^2$ support size is sufficient, and~\cite{LMM03} point
out that the general method works for any constant number $k$ of players.
It follows that the brute-force ``support enumeration'' algorithm
takes time $O(n^{\log n})$ for any $\epsilon>0$.

Very recently, H\'{e}mon et al.~\cite{HRS08} studied independently from
our work the problem of computing approximate Nash equilibria in games
with more than two players. The results they obtain are very similar to
the results presented in this paper. In particular, they also show that
for games with $k$ players a $(1-\frac{1}{k})$-Nash equilibrium with
constant support size can be computed efficiently and that this is the
best possible which can be achieved with constant support strategies.
Additionally, they also show that support size $O(\log(n)/\epsilon^2)$
is sufficient for computing $\epsilon$-approximate Nash equilibria in the
additive and multiplicative sense.

\subsection{Our Results}
We give a simple algorithm, essentially an extension of~\cite{DMP06}
from the 2-player case to the $k$-player case, that finds
$(1-\frac{1}{k})$-Nash equilibria in which each player's mixed
strategy has support size at most 2. Notice that this result becomes
quite weak as $k$ increases, given that any set of mixed strategies
constitutes a 1-Nash equilibrium (recall that we assume all payoffs
lie in the range $[0,1]$). However, we also show that one can
do no better for constant support strategies. The argument is a kind
of generalisation of the one of~\cite{FNS07} that gives the result for
the 2-player case.

For unrestricted strategies, we give a method that uses any algorithm
for the 2-player case as a component, and provides slightly better
values of $\epsilon$ than the above. Finally,
we show how the bounded differences inequality can be used to give a
simplified proof of a result of~\cite{LMM03}, that for a constant
number of players, $O(\log n)$ support size is sufficient for finding
approximate equilibria.  This allows us to deduce that if the number
of players is less than $\sqrt{n}$, there is a subexponential
algorithm for finding approximate Nash equilibria.

\section{Details of Results}
\label{sec:results}

\subsection{Solutions with Constant Support}

In this section we generalise the results of~\cite{DMP06,FNS07} from
$\epsilon=\frac{1}{2}$ in the 2-player case, to $\epsilon=1-\frac{1}{k}$
in the $k$-player case.

\begin{definition}
Define a {\em winner-takes-all game} to be one in which, for any
combination of pure strategies, one player obtains a payoff of 1, and
the other players obtain payoffs of 0.
\end{definition}

In the 2-player case, winner-takes-all games are win-lose games where
the payoffs sum to 1. Our generalisation of the lower bound
of~\cite{FNS07} follows their approach in that we generate a random
winner-takes-all game and show that for support sizes limited by some
constant $\supsize$, for any $\epsilon < 1-\frac{1}{\numplayers}$, if
the game is large enough, there will be positive probability that any
solution with supports at most $\supsize$ is not $\epsilon$-Nash.

\begin{definition}
Define the {\em total support} of a mixed-strat\-egy profile to
be the sum over all players $p$, of the number of strategies of $p$
that $p$ assigns non-zero probability.
\end{definition}

\begin{theorem}
\label{lowerboundsupport}
Let $\supsize$ and $\numplayers$ be any positive integers, and let
$\epsilon=1-\frac{1}{\numplayers}$. There exist games having
$\numplayers$ players such that in any $\epsilon$-Nash equilibrium,
there must be at least one player whose mixed strategy has support
greater than $\supsize$. (In particular, there exist games with $n$
strategies for each player such that in any $\epsilon$-NE at least one
player must have a strategy with support $\Omega(\sqrt[k-1]{\log n})$.)
\end{theorem}

\begin{proof}
We construct a $\numplayers$-player winner-takes-all game ${\cal G}$
uniformly at random as follows. Let each player have $n$ pure
strategies (a suitable value of $n$ will be identified later). For
each combination of pure strategies, choose one of the $\numplayers$
players uniformly at random and let that player have an associated
payoff of 1, while the remaining players have payoffs of 0.

We prove that for large enough $n$, with positive probability,
the resulting game has the desired property.

Consider strategy profiles with total support $\supsize$.  There are
less than $\binom{\numplayers n}{\supsize}$ possible support sets.

We choose $n$ large enough such that for any individual support set of
total size $\supsize$, the probability (with respect to random
generation of ${\cal G}$) that all players have regret less than
$\epsilon$ is less than $1/\binom{\numplayers n}{\supsize}$, so that
we can apply a union bound.

Given any mixed-strategy profile, we know that there exists at least
one player who has an expected payoff of at most $1/\numplayers$,
since the payoffs always sum to 1. Fix a support set $S$ of total size
$\supsize$.  We show that all players (including in particular the one
with lowest expected payoff) have (with positive probability) a pure
strategy that, if they unilaterally defected to it, would give them
payoff 1 (for any mixed profile using $S$).

For each player $p$, there are $n$ pure strategies to choose from.  In
order for some given strategy $s^p$ (available to player $p$) to have
expected payoff of 1 with respect to the other players' strategies, a
sufficient condition is that the payoff entries corresponding to $s^p$
and the $\leq \supsize-1$ strategies in use by the other players,
should let player $p$ win. There are fewer than
$\supsize^{\numplayers-1}$ pure-strategy combinations that the other
players can select with non-zero probability. 

The probability that $s^p$ wins against a fixed pure-strategy 
combination is $1/\numplayers$, so the probability that $s^p$ wins 
against all pure-strategy combinations of the other players is 
$\numplayers^{-(\supsize^{\numplayers-1})}$. Hence, the probability 
that player $p$ has no strategy that guarantees him a payoff of $1$ is
\[
   \left(1-\numplayers^{-(\supsize^{\numplayers-1})}\right)^{n}\enspace.
\]

The probability that there exists a player $p$ without a winning
strategy is at most
\[
  \numplayers\cdot\left(1-\numplayers^{-(\supsize^{\numplayers-1})}\right)^{n}\enspace.
\]

The probability that there exists a support set for which there exists a
player $p$ without a winning strategy is at most
\[
  {\numplayers n \choose \supsize}\cdot
  \numplayers\cdot\left(1-\numplayers^{-(\supsize^{\numplayers-1})}\right)^{n}\enspace. 
\]
We can upper bound this probability by
\[
  \numplayers\cdot(\numplayers 
  n)^\supsize\cdot\left(1-\numplayers^{-(\supsize^{\numplayers-1})}\right)^{n}\enspace.
\]

Let us set $t=\sqrt[k-1]{a\log_k(n)}$ for some constant $a>0$. Then the 
term simplifies to
\begin{equation}\label{eqn:1}
   k\cdot(kn)^{\sqrt[k-1]{a\log_k(n)}}\cdot\left(1-n^{-a}\right)^{n}\enspace.
\end{equation}
For $a=1/2$, we can estimate $\left(1-n^{-a}\right)^{n}$ by 
$e^{-\sqrt{n}}$ and hence \eqref{eqn:1} tends to zero when $n$ tends to 
infinity.

This ensures that if the size of the support set is bounded from above 
by $\sqrt[k-1]{\log_k(n)/2}$, then, for large enough $n$, with positive 
probability each player has a pure strategy with payoff 1, regardless 
of the choice of the support set.

Hence, we require a total support size of 
$\Omega(\sqrt[k-1]{\log n})$ to ensure an approximation performance of 
better than $1-1/\numplayers$, which for two players agrees with the 
upper bound of~\cite{LMM03}.
\end{proof}

\begin{corollary}
If we restrict strategies to have constant support, there is
a lower bound of $1-\frac{1}{k}$ on the approximation quality
that we can guarantee.
\end{corollary}

Note that the above is a generalisation of a lower bound
of~\cite{A94,FNS07} from the 2-player case to the $k$-player case.
We show that this lower bound is optimal by giving an algorithm
(which is an extension of an algorithm of~\cite{DMP06})
that finds a $1-\frac{1}{k}$-approximate Nash equilibrium
where each player has support size at most 2.

\begin{theorem}
Let ${\cal G}$ be a $k$-player game with $n$ strategies per
player. Let $\epsilon=1-\frac{1}{k}$.  There is a mixed strategy
profile in which each player has support at most 2 which constitutes
an $\epsilon$-Nash equilibrium.
\end{theorem}

\begin{proof}
The following algorithm achieves the stated objective.
\begin{enumerate}
\item
Number the players $1,\ldots,k$. For $i=1,\ldots,k-1$ in
increasing order, player $i$ allocates a probability
of $1-\frac{1}{k+1-i}$ to an arbitrary pure strategy $s_i$.
\item
Player $k$ allocates probability 1 to a pure best response $b_k$ to
the strategy combination $(s_1,\ldots,s_{k-1})$ of players
$1,\ldots,k-1$.
\item
For $i=k-1,\ldots,1$ in decreasing order, player $i$ allocates his
remaining probability $\frac{1}{k+1-i}$ to a pure best response $b_i$
to the mixed strategy formed so far, i.e., to the strategy
$(s_1,\ldots,s_{i-1},r_{i+1},\ldots,r_k)$, where $r_j$ denotes the
mixed strategy $(1-\frac{1}{k+1-j})\cdot s_j+(\frac{1}{k+1-j})\cdot
b_j$.
\end{enumerate}

In order to prove the theorem, we compute the regret of a player 
$i\in\{1,\ldots,k\}$. To avoid case distinctions, we denote by $s_k$ an
arbitrary pure strategy of player $k$. If player $i$ plays strategy
$s_i$, which happens with probability $1-\frac{1}{k+1-i}$, his regret can
only be bounded by $1$. If player $i$ plays strategy $b_i$, then his
regret is $0$ if all players $j<i$ play strategy $s_j$ and his regret
can be as bad as $1$ otherwise. Altogether, this implies that the regret
of player $i$ is bounded by
\begin{align*}
  & \Pr{\mbox{$i$ plays $s_i$}}+\Pr{\mbox{$i$ plays $b_i$}}\cdot\Pr{\exists j<i: \mbox{$j$ plays $b_j$}}\\
  =& \left(1-\frac{1}{k+1-i}\right) + \frac{1}{k+1-i}\left(1-\prod_{j=1}^{i-1}\left(1-\frac{1}{k+1-j}\right)\right)\\
  =& 1 - \frac{1}{k}\enspace. \qedhere
\end{align*}
\end{proof}

\subsection{Solutions with non-constant support}

The general issue of interest is the question of what approximation
guarantee can we obtain for a polynomial-time algorithm for $k$-player
approximate Nash equilibrium. To what extent can we improve on the
above (weak) result when we allow unrestricted mixed strategies? We do
not have a very substantial improvement, but the following shows how
the trick of~\cite{DMP06} can be combined with better 2-player
algorithms to get an improvement over the $(1-\frac{1}{k})$-result
that the above would yield for $k$ players.

\begin{theorem}
If 2-player $\epsilon$-approximate Nash equilibria can be found in 
polynomial time, then $k$-player 
$\frac{(k-2)-(k-3)\epsilon}{(k-1)-(k-2)\epsilon}$-Nash equilibria can 
be found in polynomial time.
\end{theorem}

As a corollary, the approximation guarantee of $0.3393$
associated with the algorithm of~\cite{TS07} gives a 3-player
algorithm with an approximation guarantee of $0.6022$.

\begin{proof}
We prove the theorem by induction on the number of players. Let 
$\delta_k=\frac{(k-2)-(k-3)\epsilon}{(k-1)-(k-2)\epsilon}$. For $k=2$, 
which is the induction basis, $\delta_2=\epsilon$. Now assume that 
there is an algorithm ${\cal A}_{k-1}$ with approximation guarantee 
$\delta_{k-1}$ for $(k-1)$-player games. The following algorithm for
$k$-player games achieves the stated approximation guarantee $\delta_k$.

\begin{enumerate}
\item Player 1 allocates a probability of $1/(2-\delta_{k-1})$ to some
arbitrary pure strategy $s$.
\item Players $2,\ldots,k$ apply ${\cal A}_{k-1}$ to the $(k-1)$-player game
that results from letting player 1 play $s$.
\item Player 1 allocates his remaining probability to a pure best
response $b$ to the strategies of players $2,\ldots,k$ chosen in step 2.
\end{enumerate}

A player $i\in\{2,\ldots k\}$ has regret at most $\delta_{k-1}$ if player 1
plays strategy $s$ and his regret can be as bad as 1 otherwise. Hence,
the regret of a player $i\in\{2,\ldots k\}$ can be bounded by
\[
  \left(\frac{1}{2-\delta_{k-1}}\right)\cdot\delta_{k-1}+\left(1-\frac{1}{2-\delta_{k-1}}\right)
  = \frac{1}{2-\delta_{k-1}}\enspace.
\]
Player 1 has no regret when playing $b$ and his regret can be as bad as
1 when playing $s$. This implies that also the regret of player 1 can be
bounded by $1/(2-\delta_{k-1})$. A simple calculation shows
\[
   \frac{1}{2-\delta_{k-1}} =
   \frac{1}{2-\frac{(k-3)-(k-4)\epsilon}{(k-2)-(k-3)\epsilon}} 
   = \frac{(k-2)-(k-3)\epsilon}{(k-1)-(k-2)\epsilon}
   = \delta_k\enspace,
\]
as desired.
\end{proof}

Finally, we use the bounded differences inequality to give a
simplified proof of a result of~\cite{LMM03}.

\begin{theorem}
Let ${\cal G}$ be a $k$-player game in which each player has
$n$ pure strategies. Let $\epsilon>0$.
Then there exists a mixed strategy profile
where each player has support $k^2\log(kn)/2\epsilon^2$ that
constitutes an $\epsilon$-Nash equilibrium.
In particular, the players' distributions are empirical
distributions over multisets of size $k^2\log(kn)/2\epsilon^2$.
\end{theorem}

\begin{proof}
The bounded differences inequality (see e.g.~\cite{DL01} p.~8) is the
following.  Let $A$ be a set.
Let $g:A^N \longrightarrow {\bf R}$ be a function with the
{\em bounded difference} property: for all $i\in\{1,\ldots,n\}$ and for
all $x_1,\ldots,x_N,x'_i\in A$
\[
| g(x_1,\ldots,x_N)-g(x_1,\ldots,x_{i-1},x'_i,x_{i+1},\ldots,x_N) |
\leq c_i.
\]
Suppose $X_1,\ldots,X_N$ are independent random variables over $A$.
The bounded difference inequality states that for all $t>0$,
\begin{align*}
& \quad\, \Pr{g(X_1,\ldots,X_N)-{\bf E} g(X_1,\ldots,X_N) \geq t}\\
& \leq \exp\left(\frac{-2t^2}{\sum_{i=1}^N c^2_i}\right),
\end{align*}
and
\begin{align*}
& \quad\, \Pr{{\bf E} g(X_1,\ldots,X_N) - g(X_1,\ldots,X_N) \geq t}\\
& \leq \exp\left(\frac{-2t^2}{\sum_{i=1}^N c^2_i}\right).
\end{align*}
We apply the above as follows. Let $A=[n]^k$, the set of all pure
profiles of $k$-person games where each player has $n$ strategies.
The $X_i$ are samples from some fixed Nash equilibrium.
Let $g_{i,j}$ be the payoff to player $i$ for using strategy $j$
subject to all players using the mixture of $N$ strategies
obtained by taking, for each player $p$, the $p$-th entries of each
$X_m$ ($m\in [N]$), and using the uniform distribution on that multiset.

Note that for the $g_{i,j}$ functions, $c_i \leq k/N$.

We want the right-hand sides to be less than $1/2kn$, so that by a
union bound, for each $i$ and $j$, we have the payoff for player $i$
using strategy $j$ being close to expected.  Thus, we want $N$ such
that
\[
\exp(-2\epsilon^2/(N(k^2/N^2))) \leq \frac{1}{2kn}.
\]
Solving for $N$ we get
\[
N \geq \frac{k^2 \log (2kn)}{2\epsilon^2}.\qedhere
\]
\end{proof}

The theorem indicates that support enumeration is a subexponential
algorithm provided that $k=o(n^{1/2})$ (a brute-force search
has complexity $O(n^{kN})$ which is subexponential for
$k=O(n^{1/2-\zeta})$.) Let us remark, that for $k=o(n^{1/3})$
support enumeration is not only subexponential in the input size $n^k$ but also
in the number $n$ of different strategies of the players. Hence, in
this case, the running time is subexponential even for more succinctly
represented games.

\section{Conclusions}

It seems that only very weak approximation performance is possible
with constant-support strategies. Work on the special case
of 2-player games has mostly addressed this problem by constructing
linear programs whose solutions have useful properties (and are
typically non-constant support strategies). For more than 2 players
this approach no longer works; the corresponding constraints are
no longer linear. The challenge seems to be to find alternative
ways of describing mixed strategies (of more than constant
support) that have desirable properties.

\end{document}